\documentclass{article}

\usepackage{arxiv}

\usepackage[utf8]{inputenc} % allow utf-8 input
\usepackage[T1]{fontenc}    % use 8-bit T1 fonts
\usepackage{hyperref}       % hyperlinks
\usepackage{url}            % simple URL typesetting
\usepackage{booktabs}       % professional-quality tables
\usepackage{amsfonts}       % blackboard math symbols
\usepackage{nicefrac}       % compact symbols for 1/2, etc.
\usepackage{microtype}      % microtypography
\usepackage{lipsum}         % Can be removed after putting your text content
\usepackage{graphicx}
\usepackage{doi}

\usepackage[utf8]{inputenc} % allow utf-8 input
\usepackage[T1]{fontenc}    % use 8-bit T1 fonts
\usepackage{booktabs}       % professional-quality tables
\usepackage{microtype}      % microtypography
\usepackage{wrapfig}

\usepackage{latexsym}
\usepackage[framemethod=TikZ]{mdframed}
\usepackage{amsmath,mathtools,amsfonts,mathrsfs,amssymb,bm}
\usepackage{amsthm}
\usepackage{stmaryrd}
\usepackage{nicefrac}
\DeclareMathAlphabet{\mathbbm}{U}{bbm}{m}{n}
\usepackage{xspace}
\usepackage{paralist}
\usepackage{enumitem}
\usepackage{picinpar}
\usepackage{xcolor,colortbl}
\usepackage{placeins}
\usepackage{hyperref}
\hypersetup{
    colorlinks=true,
    linkcolor=blue,
    citecolor=blue,
    urlcolor=blue,
}
\usepackage{subcaption,cleveref}

\def\rrr#1\\{\par
\medskip\hbox{\vbox{\parindent=2em\hsize=6.12in
\hangindent=4em\hangafter=1#1}}}

\newmdtheoremenv[innertopmargin=-2pt]{research}{Research direction}
\mdfdefinestyle{objstyle}{
    linecolor=black,            
    linewidth=1pt,       
    frametitlerule=true,
    backgroundcolor=blue!05,
}

\newcommand{\minimize}{\operatornamewithlimits{minimize}}

\newtheorem{theorem}{{\bf Theorem}}

\newtheorem{definition}{{\bf Definition}}

\newtheorem{proposition}{{\bf Proposition}}

\definecolor{maincolor}{HTML}{032F99}
\definecolor{blue}{RGB}{31,64,122}
\definecolor{red}{HTML}{e05a87}

\title{Optimal Allocation of Privacy Budget \\ on Hierarchical Data Release}

\date{}

\newif\ifuniqueAffiliation
\uniqueAffiliationtrue

\ifuniqueAffiliation % Standard variant of author block
\author{%
  Joonhyuk Ko \\
  University of Virginia \\
  \texttt{tah3af@virginia.edu} \\
  \And
  Juba Ziani \\
  Georgia Institute of Technology \\
  \texttt{jziani3@gatech.edu} \\
  \And
  Ferdinando Fioretto \\
  University of Virginia \\
  \texttt{fioretto@virginia.edu} \\
}
\else
\usepackage{authblk}

\newbox{\orcid}\sbox{\orcid}{\includegraphics[scale=0.06]{orcid.pdf}} 
\author[1]{%
	\href{https://orcid.org/0000-0000-0000-0000}{\usebox{\orcid}\hspace{1mm}David S.~Hippocampus\thanks{\texttt{hippo@cs.cranberry-lemon.edu}}}%
}
\author[1,2]{%
	\href{https://orcid.org/0000-0000-0000-0000}{\usebox{\orcid}\hspace{1mm}Elias D.~Striatum\thanks{\texttt{stariate@ee.mount-sheikh.edu}}}%
}
\affil[1]{Department of Computer Science, Cranberry-Lemon University, Pittsburgh, PA 15213}
\affil[2]{Department of Electrical Engineering, Mount-Sheikh University, Santa Narimana, Levand}
\fi

\renewcommand{\undertitle}{}
\renewcommand{\shorttitle}{Optimal Allocation of Privacy Budget on Hierarchical Data Release}

\hypersetup{
pdftitle={A template for the arxiv style},
pdfsubject={q-bio.NC, q-bio.QM},
pdfauthor={David S.~Hippocampus, Elias D.~Striatum},
pdfkeywords={First keyword, Second keyword, More},
}

\begin{document}
\maketitle

\begin{abstract}
	Releasing useful information from datasets with hierarchical structures while preserving individual privacy presents a significant challenge. Standard privacy-preserving mechanisms, and in particular Differential Privacy, often require careful allocation of a finite privacy budget across different levels and components of the hierarchy. Sub-optimal allocation can lead to either excessive noise, rendering the data useless, or to insufficient protections for sensitive information. This paper \emph{addresses the critical problem of optimal privacy budget allocation for hierarchical data release}. It formulates this challenge as a constrained optimization problem, aiming to maximize data utility subject to a total privacy budget while considering the inherent trade-offs between data granularity and privacy loss. The proposed approach is supported by theoretical analysis and validated through comprehensive experiments on real hierarchical datasets. These experiments demonstrate that optimal privacy budget allocation significantly enhances the utility of the released data and improves the performance of downstream tasks.
\end{abstract}

% keywords can be removed
% \keywords{First keyword \and Second keyword \and More}

\section{Introduction}\label{intro}

Hierarchical data structures are ubiquitous, appearing in domains such as geographical information systems, organizational charts, biological taxonomies, and product catalogs. While the analysis of such data impacts society and the economy as it touches disciplines from scientific discovery to policy-making, it is often regulated by privacy laws or ethical considerations. For instance, the U.S. Census Bureau's data is a prime example of hierarchical data, where population statistics are organized at various levels, such as national, state, and county. This data is crucial for a wide range of applications, including resource allocation, migration studies, and public health decisions. Its release is also governed by strict privacy regulations, including the Title 13 of the U.S.~Code, which prohibits the disclosure of personally identifiable information. As a result, privacy-preserving techniques have been developed to ensure that sensitive information remains confidential while still allowing for meaningful analysis. Among these, Differential Privacy (DP) \cite{DP_def} has emerged as the gold standard, primarily due to its strong, quantifiable guarantees against re-identification and inference attacks.

However, hierarchical data presents unique challenges for differential privacy, as its privacy budget must be allocated across different levels of the hierarchy. 
For instance, one might need to release aggregate statistics at a coarse geographical level (e.g., state) and more fine-grained statistics at a lower level (e.g., county or city). A naive or uniform allocation of the privacy budget across these levels can lead to suboptimal outcomes: either too much noise is added to fine-grained data, destroying its utility, or coarse-grained data is overly protected at the expense of detail elsewhere. This necessitates a principled approach to privacy budget allocation that explicitly considers the structure of the data and the utility requirements at different granularities.

This observation raises a fundamental question: {\em How can we optimally allocate a given privacy budget across a hierarchical data structure to maximize the utility of the released information?} The answer to this question is not straightforward, as it involves balancing the trade-off between privacy and utility across multiple levels of the hierarchy. As the paper will show, commonly adopted heuristic approaches often achieve suboptimal results, leading to inaccurate downstream analyses. 
This paper precisely addresses this problem, and provides the following contributions:
\begin{enumerate}[label={$\bm{C}_\arabic*$}, leftmargin=*, itemsep=0pt, parsep=0pt, topsep=-2pt]
    \item It defines and characterizes the problem of optimal privacy budget allocation for hierarchical data release under differential privacy. This includes defining appropriate utility metrics that capture the usefulness of data at different levels of the hierarchy.
    \item It provides an analysis of the properties of the optimal allocation strategy (as detailed in Section~\ref{Opt_Alloc_Section}), in particular the relationship between bias and variance of the released data, and how these are affected by the privacy budget allocation. 
    \item Crucially, the analysis shows that the optimal allocation strategy is not necessarily uniform across levels, and utility can be significantly improved by accounting for the hierarchy's structure. To address this, it proposes an optimization-based approach that minimizes an analytical expression for mean squared error, combining bias and variance introduced by differential privacy under non-negativity post-processing, commonly used in histogram and contingency table releases. Importantly, by minimizing this objective, the method improves the utility of the released data even after the post-processing, which is known to distort bias and variance~\cite{ZHF:aaai21}.
    % \item Crucially, the analysis reveals that the optimal allocation strategy is not necessarily uniform across levels, and that the utility of the released data can be significantly improved by carefully considering the structure of the hierarchy.
    % Given these considerations, it proposes an optimization-based approach to solve this allocation problem. The proposed approach leverages the analytical expressions for the bias and variance introduced by differential privacy under non-negativity post-processing, commonly adopted in the release of 
    % data such as histograms and contingency tables. %, and combines them into a mean squared error objective. 
    % Importantly, by minimizing the resulting objective, the method improves the utility of the released data even after the post-processing, which is known to distort bias and variance~\cite{ZHF:aaai21}.
    \item Finally, the paper demonstrates the practical utility of the proposed method through extensive experiments on real-world datasets derived from the U.S.~Census. The results show a substantial reduction of \textbf{both} bias and variance compared to uniform allocation under the same privacy budget. 
    Furthermore, motivated by U.S.~Census budget allocation, the framework is extended to a downstream resource allocation task using privatized data. This extension examines how different classes of preference functions influence bias and variance, with results consistently showing higher utility for the optimal allocation strategy over uniform allocation.
\end{enumerate}
The remainder of this paper is organized as follows. Section~\ref{Prelim} introduces essential preliminaries on differential privacy and hierarchical data, and sets up the problem. Section~\ref{Opt_Alloc_Section} details our proposed allocation methodology.
Section~\ref{Exp_Section} presents the experimental results, focusing on Wyoming as a running example, with additional state experiments in Appendix~\ref{additional_experiments_State}.
Section~\ref{Downstream_Section} explores the impact on downstream tasks. Finally, Section~\ref{Conclusion} concludes the paper and outlines future research directions.

\begin{figure}[t]
    \centering
    \begin{minipage}[h]{0.60\textwidth}
        \centering
        \includegraphics[width=0.78\textwidth]{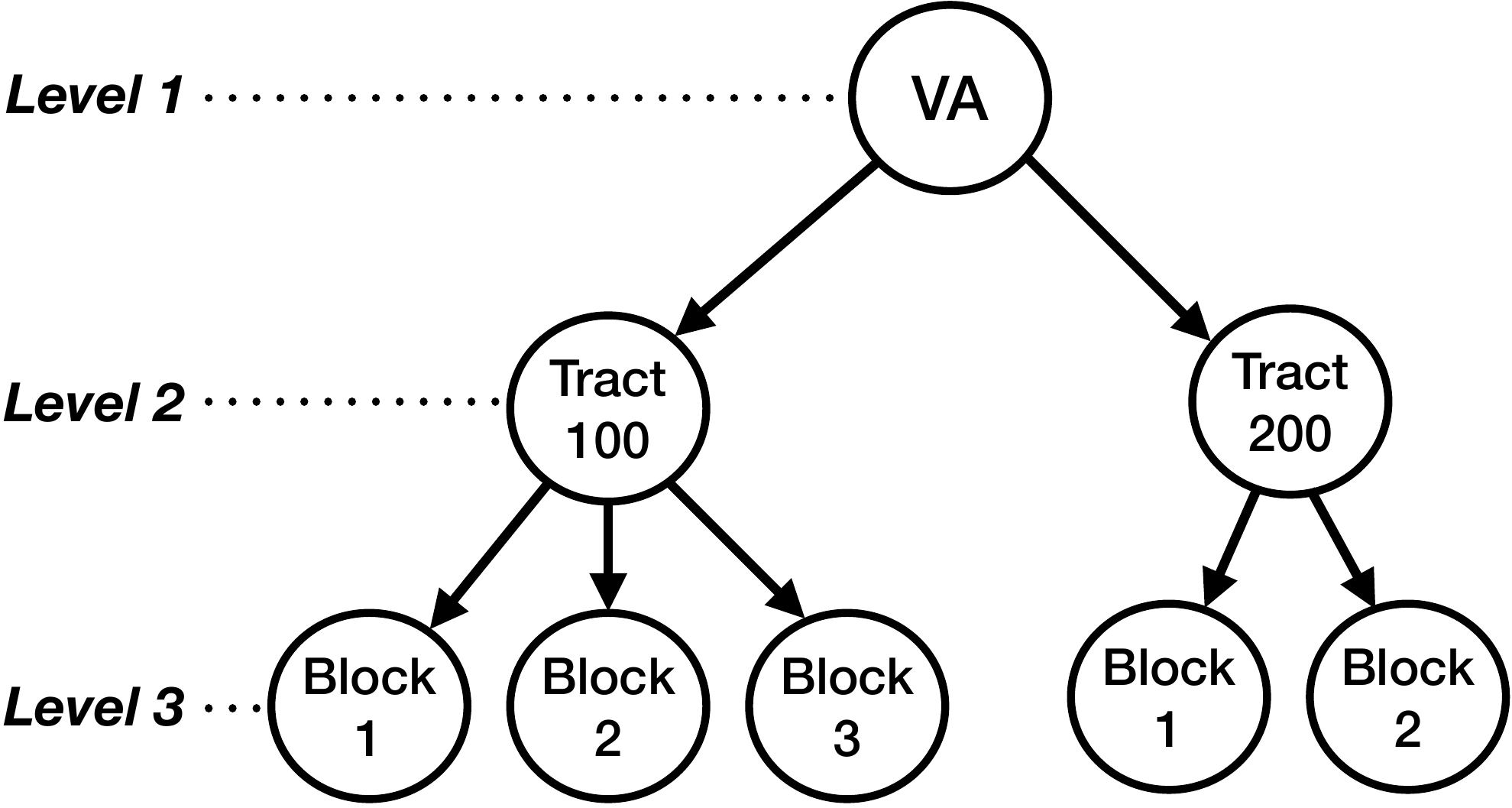}
        \captionof{figure}{A hierarchical structure of geographic regions, where blocks are nested within tracts, and tracts within the state of Virginia. }
        \label{fig:hierarchy_figure}
    \end{minipage}
    \hfill
    \begin{minipage}[h]{0.38\textwidth}
        \centering
        \small % Reduce font size to fit better
        \begin{tabular}{@{}llr@{}}
            \toprule
            \textbf{Level} & \textbf{Region} & \textbf{Population} \\
            \midrule
            Level 3 & Block 1 (100) & 120 \\
            Level 3 & Block 2 (100) & 80 \\
            Level 3 & Block 3 (100) & 100 \\
            \addlinespace
            Level 2 & \textbf{Tract 100} & \textbf{300} \\
            \midrule
            Level 3 & Block 1 (200) & 90 \\
            Level 3 & Block 2 (200) & 60 \\
            \addlinespace
            Level 2 & \textbf{Tract 200} & \textbf{150} \\
            \midrule
            Level 1 & \textbf{VA} & \textbf{450} \\
            \bottomrule
        \end{tabular}
        \captionof{table}{Population counts corresponding to each region in the hierarchy shown in Figure~\ref{fig:hierarchy_figure}.}
        \label{tab:hierarchy_figure_table}
    \end{minipage}
\end{figure}

\paragraph{Related Work.}\label{Related_work}
Since the adoption of the U.S. Census Bureau's TopDown Algorithm (TDA)~\cite{USCB_adopts_DP}, there has been growing interest in exploring hierarchical data release under differential privacy. This has motivated theoretical and practical approaches for allocating privacy budgets across hierarchical levels~\cite{boninsegna_2025_DP_release_hierarchical, FvHZ:AIJ21, Cohen_DP_Redistricting_2021, Cohen2022Private}. Accurately privatizing hierarchical data releases, such as those in the U.S.~Census, is crucial, as it directly impacts federal resource allocation and has significant implications for fairness in downstream decisions~\cite{pujol_fair_decision, Wu2022policy, Fioretto_decision_making_under_fairness}. 
Beyond population statistics, hierarchical data plays a key role in other domains. Liu et al~\cite{liu_2022_private_synthetic_data_hierarchical} explored generating private synthetic hierarchical data, while hierarchical structures have also been used in federated learning to enhance privacy and scalability~\cite{chandrasekaran2022hierarchicalfederatedlearningprivacy}.
Fioretto et al.~\cite{FvHZ:AIJ21} proposed an optimization-based method for post-processing privatized hierarchical data to ensure consistency and improve utility. Cohen et al.~\cite{Cohen_DP_Redistricting_2021} examined how privacy budget allocation affects redistricting accuracy in the Census TopDown algorithm, finding that bottom-heavy allocations improve accuracy for small geographic units. Similarly, Dawson et al.~\cite{dawson2023} proposed a greedy iterative approach using noisy priors and demonstrated that unequal budget allocations can achieve higher accuracy than uniform splits.

However, there remains a gap in understanding how to optimally allocate privacy budgets across different levels of the hierarchy, especially when common post-processing alters the bias and variance of differentially private estimates.

\section{Preliminaries and Problem Setting}\label{Prelim}

\subsection{Hierarchical Data Release}

This paper considers the hierarchical release of population statistics. Let \( R_1, R_2, \ldots, R_L \) denote the sets of regions at each level of a hierarchy, where each region in \( R_\ell \) is indexed by a tuple \( (i_1, \ldots, i_\ell) \) indicating its position in the hierarchy. Specifically, \( R_1 \) contains highest-level regions (e.g., states in Fig~\ref{fig:hierarchy_figure}), and each \( R_\ell \) contains subregions nested within regions from level \( \ell - 1 \). Let \( N_{i_1, \ldots, i_\ell} \) denote the population count associated with region \( (i_1, \ldots, i_\ell) \in R_\ell \). These counts satisfy the consistency constraints across levels:
\[
\sum_{(i_1, \ldots, i_L) \in R_L} N_{i_1, \ldots, i_L}
= \cdots = 
\sum_{i_1 \in R_1} N_{i_1} = N_\text{total}.
\]

\textbf{$\blacklozenge$ Example.} In Figure~\ref{fig:hierarchy_figure}, \(R_1 = \{(\text{VA})\}\) with \(N_\text{total} = 450\). At level 2, we have \(R_2 = \{(\text{VA}, 100), (\text{VA}, 200)\}\) with populations \(N_{\text{VA}, 100} = 300\) and \(N_{\text{VA}, 200} = 150\). At level 3, \(R_3\) contains five blocks: \((\text{VA}, 100, 1), (\text{VA}, 100, 2), (\text{VA}, 100, 3), (\text{VA}, 200, 1), (\text{VA}, 200, 2)\) with respective populations 120, 80, 100, 90, and 60. These satisfy the consistency condition:
\[
\sum_{(i_1, i_2, i_3) \in R_3} N_{i_1, i_2, i_3}
= \sum_{(i_1, i_2) \in R_2} N_{i_1, i_2}
= N_\text{total}.
\]

\subsection{Differential Privacy}
This paper considers the problem of releasing privatized hierarchical data under the framework of differential privacy (DP)~\cite{DP_def, DP_book}. Differential privacy provides a formal framework for quantifying the privacy loss incurred when releasing information derived from sensitive data. It ensures that the output of a computation is nearly indistinguishable whether or not any single individual's data is included in the input dataset. This is typically achieved by adding controlled noise to the output, calibrated to a privacy budget that quantifies the level of privacy protection. 

\begin{definition}[Differential Privacy]
    A randomized algorithm $\mathcal{M}$ satisfies $\varepsilon$-differential privacy if, for all $\mathcal{S} \subseteq \operatorname{Range}(\mathcal{M})$ and any two neighboring databases  $D_1 \sim D_2$ (i.e., differing by one entry),
    \[
    \Pr[\mathcal{M}(D_1) \in \mathcal{S}] \leq \exp(\varepsilon) \Pr[\mathcal{M}(D_2) \in \mathcal{S}].
    \]
\end{definition}

Here, $\varepsilon > 0$ is the privacy parameter that quantifies the privacy loss: smaller values ($\varepsilon \to 0$) imply strong privacy, while larger values ($\varepsilon \to \infty$) correspond to weak privacy.

Differential privacy provides two key properties essential to this work: sequential composition and post-processing immunity. Sequential composition quantifies total privacy loss across hierarchical levels, while post-processing immunity guarantees that adjustments such as enforcing consistency or non-negativity do not affect the privacy guarantee.

\begin{theorem}[Sequential Composition]
    Let $\mathcal{M}_1$ be an $\varepsilon_1$-differentially private algorithm, and let $\mathcal{M}_2$ be an $\varepsilon_2$-differentially private algorithm. Then the combined algorithm that outputs $(\mathcal{M}_1(x), \mathcal{M}_2(x))$ is $\varepsilon_1 + \varepsilon_2$-differentially private.
\end{theorem}

\begin{theorem}[Post-Processing Immunity]
    If a randomized algorithm $\mathcal{M}$ is $\varepsilon$-differentially private, and $f$ is any arbitrary function, then the composition $f \circ \mathcal{M}$ is also $\varepsilon$-differentially private.
\end{theorem}

\subsection{Differentially Private Hierarchical Data Release Model}
Motivated by Census applications, we consider the release of privatized hierarchical data by adding noise at each level of a hierarchy. Let $\varepsilon_1, \varepsilon_2, \ldots, \varepsilon_L$ denote the privacy budgets allocated from the highest level ($\varepsilon_1$) to the lowest ($\varepsilon_L$), with the total privacy budget given by the sequential composition:
\[
\varepsilon_\text{Total} = \varepsilon_1 + \varepsilon_2 + \ldots + \varepsilon_L.
\]

In this work, we add independent Laplace\footnote{While we focus on the Laplace mechanism, similar analysis can be carried out under other mechanisms such as Gaussian, with qualitatively similar tradeoffs between bias, variance, and utility.} noise (which is chosen for clarity in the analysis of bias and variance later in the paper) to each node in the hierarchy, and apply non-negativity post-processing to ensure that released counts remain valid for population or frequency data, as noise can render some counts negative:
\begin{equation}
\tilde{N} = \max\left(0, N + \text{Lap}\left(\frac{1}{\varepsilon}\right)\right). \label{non-neg_post-processing}
\end{equation}
Independent noise addition can lead to inconsistencies across levels,  such as lower-level counts not summing to their corresponding higher-level total. To address this, we apply a hierarchical consistency post-processing step, also described in Section~\ref{Opt_Alloc_Section}.

\section{The Optimal Allocation Problem}\label{Opt_Alloc_Section}

Given a fixed total privacy budget, determining how to allocate noise across hierarchical levels is non-trivial. On one hand, assigning more of the privacy budget to lower levels can improve the accuracy of fine-grained counts, which may be beneficial for downstream tasks. On the other hand, allocating more of the budget to higher levels may be preferable when hierarchical consistency constraints are imposed. Even if we understand whether a bottom-heavy or top-heavy allocation is preferable, it remains unclear how much budget should be assigned to each level.

\subsection{Bias and Variance of $\tilde{N}$}
A high-utility privatized release should minimize the mean squared error (MSE), defined as:
\begin{align}
\mathrm{MSE} = \mathbb{E}[(\tilde{N} - N)^2] = \mathrm{Bias}(\tilde{N})^2 + \mathrm{Var}(\tilde{N})  \label{eq:MSE},
\end{align}
which decomposes into a bias and a variance term through the standard decomposition~\cite{bishop2006pattern}. Thus, a high-utility hierarchical release should exhibit low bias and variance. We begin by deriving the bias introduced by the non-negativity post-processing applied after adding Laplace noise, as shown in prior work~\cite{Ko_AAAI_2025}.
\begin{proposition}
\label{proposition:bias}
For the population $N$, the bias of $\tilde{N}$ is given in closed-form by:
\[
\mathrm{Bias}(\tilde{N}) = 
\mathbb{E} \left[\tilde{N} \right] - N = 
\frac{1}{2 \varepsilon} \exp \left(-N \varepsilon\right) > 0.
\]
\end{proposition}
\begin{proof}
    The proof is provided in Appendix~\ref{appendix:proofs}.
\end{proof}
% \begin{proof}[Proof sketch]
% By computing $\mathbb{E}[\tilde{N}]$ via integration over the nonnegative region, we obtain the closed-form expression $\mathbb{E}[\tilde{N}] = N + \frac{1}{2\varepsilon} \exp(-\varepsilon N)$.
Note that the bias is strictly positive. As $\varepsilon \rightarrow \infty$ (no privacy) or $N \rightarrow \infty$, the $\mathrm{Bias}(\tilde{N}) \rightarrow 0$. This implies that lower levels of the hierarchy—where each count is typically smaller—will suffer more from the positive bias.
% As a result, allocating more privacy budget to lower levels can help mitigate this effect.

We now derive the variance of the non-negativity post-processing applied after adding Laplace noise.
\begin{proposition}
\label{proposition:variance}   
For the population $N$, the variance of $\tilde{N}$ is given in closed-form by:
\[
\mathrm{Var}(\tilde{N}) =
\frac{1}{\varepsilon^2} \left(2 - e^{-N \varepsilon}\right)
- \frac{N}{\varepsilon} e^{-N \varepsilon}
- \frac{1}{4 \varepsilon^2} e^{-2 N \varepsilon}.
\]
\end{proposition}
\begin{proof}
    The proof is provided in Appendix~\ref{appendix:proofs}.
\end{proof}
% \begin{proof}[Proof sketch]
% We compute $\mathrm{Var}(\tilde{N}) = \mathbb{E}[\tilde{N}^2] - (\mathbb{E}[\tilde{N}])^2$ by splitting the second moment into two integrals over $[0, N]$ and $[N, \infty)$. Substituting the known expression for $\mathbb{E}[\tilde{N}]$ in Lemma~\ref{lemma:bias} yields the closed-form result. Full derivation is provided in Appendix~\ref{appendix:proofs}.
% \end{proof}

As $\varepsilon \rightarrow \infty$ (no privacy), the variance $\mathrm{Var}(\tilde{N}) \rightarrow 0$. Conversely, as $N \rightarrow \infty$, the variance increases monotonically and approaches $\frac{2}{\varepsilon^2}$, which is the variance of the untruncated Laplace noise. This indicates that small population counts incur lower variance due to the truncation effect introduced by the non-negativity constraint, which reduces the overall variability. However, this reduction in variance comes at the cost of increased bias, as lower truncation shifts the mean upwards.

Using the standard decomposition in equation~\eqref{eq:MSE} and applying Proposition~\ref{proposition:bias} and \ref{proposition:variance}, the MSE admits the following closed-form expression:
\begin{align}
\mathrm{MSE} = \mathrm{Bias}(\tilde{N})^2 + \mathrm{Var}(\tilde{N}) = \frac{1}{\varepsilon^2} \left( 2 - e^{-\varepsilon N} \right)
- \frac{N}{\varepsilon} e^{-\varepsilon N}.
\end{align}

As $\varepsilon \rightarrow \infty$ (no privacy), the noise magnitude vanishes and $\mathrm{MSE} \rightarrow 0$.  Conversely, as $\varepsilon \rightarrow 0$, we have $\mathrm{MSE} \rightarrow \infty$ due to increasing noise. When $N \rightarrow \infty$, the positive bias diminishes, and the MSE is dominated by the variance term. \emph{This highlights that bias is most significant when true counts are small, whereas for large counts, the error is primarily driven by noise variance.} We now state two key propositions that play an important role in optimization and privacy budget allocation.  We provide proofs in Appendix~\ref{appendix:proofs}.

\begin{proposition}
\label{proposition:MSE_convex}
    The $\mathrm{MSE}$ is strictly convex in $\varepsilon$ over $\mathbb{R}_{> 0}$ for any fixed $N > 0$.
\end{proposition}

\begin{proposition}
\label{proposition:MSE_bounded}
    For any fixed $\varepsilon > 0$ and $N \geq 0$, the $\mathrm{MSE}$ is bounded as
    \[
    \frac{1}{\varepsilon^2} \leq \mathrm{MSE} < \frac{2}{\varepsilon^2}.
    \]
\end{proposition}

In the context of hierarchical data release, the \emph{total MSE} across a multi-level hierarchy with $L$ levels (each allocated a privacy budget $\varepsilon_\ell$) is given by:
\begin{align}
\mathsf{MSE}_\text{Total} = \sum_{\ell=1}^L \sum_{j \in R_\ell} \left[
\frac{1}{\varepsilon_\ell^2}(2 - e^{-\varepsilon_\ell N_j}) - \frac{N_j}{\varepsilon_\ell} e^{-\varepsilon_\ell N_j}
\right],
\end{align}
where $R_\ell$ denotes the set of nodes at level $\ell$. This expression simply sums the per-node MSE across all levels of the hierarchy. Then, Proposition~\ref{proposition:MSE_bounded}  suggests the following: under a uniform privacy budget (i.e., $\varepsilon_1 = \cdots = \varepsilon_L$), lower levels, containing as many or more nodes than higher levels, contribute more to \(\mathsf{MSE}_\text{Total}\), as the total MSE incurred at these levels is at least as large as that of the higher level. Furthermore, since a finite sum of strictly convex functions remains strictly convex, $\mathsf{MSE}_\text{Total}$ is strictly convex in $\varepsilon$ by Proposition~\ref{proposition:MSE_convex}. This property is crucial, as it ensures that the privacy budget allocation problem can be formulated as a convex optimization problem.

\subsection{Optimization Programs}

To allow flexibility in prioritizing accuracy at different levels, we define a \emph{weighted MSE objective} that incorporates level-wise preferences through weights:
\begin{align}
\mathsf{MSE}_{\mathbf{w}}(\boldsymbol{\varepsilon}) = \sum_{\ell=1}^L w_\ell \cdot \underbrace{\sum_{j \in R_\ell} \left[
\frac{1}{\varepsilon_\ell^2}(2 - e^{-\varepsilon_\ell N_j}) - \frac{N_j}{\varepsilon_\ell} e^{-\varepsilon_\ell N_j}
\right]}_\text{MSE at level $\ell$}, \label{eq:multi_level_mse}
\end{align}
where $\mathbf{w} = (w_1, \dots, w_L) \in \mathbb{R}_{\geq 0}^L$ denotes the level-wise weights, and $\boldsymbol{\varepsilon} = (\varepsilon_1, \dots, \varepsilon_L)$ is the vector of privacy budgets. These weights allow the data curator to emphasize accuracy at specific levels based on downstream utility needs. In Section~\ref{optimized_alloc_bias_var_section}, we focus on the case of equal weights ($w_1 = w_2 = \cdots = w_L$), minimizing the total MSE across the hierarchy without prioritizing any particular level. An ablation study examining how different weight choices affect overall and per-level accuracy is then presented in Section~\ref{ablation}. We now formalize two natural optimization problems.

\textbf{(1) Minimize MSE subject to a total privacy constraint:}
\begin{subequations}
    \label{eq:1}
    \begin{align}
        \minimize_{\boldsymbol{\varepsilon} \in \mathbb{R}_{\geq 0}^L}& \quad \mathsf{MSE}_{\mathbf{w}}(\boldsymbol{\varepsilon}) \\
        \texttt{s.t.}& \quad \sum_{\ell=1}^L \varepsilon_\ell \leq \varepsilon_\text{Total}.
    \end{align}
\end{subequations}

This program reflects settings where a fixed total privacy loss must be allocated across levels to minimize total error. It is important to note that to satisfy $\varepsilon$-DP, one must base the optimization on previously released statistics or a similar data structure, rather than directly computing the total MSE from the raw statistics to be privatized. This is because using raw statistics in the algorithm to determine the privacy budget would itself leak information. However, if only the total privacy loss is released (rather than individual allocations at each hierarchical level), then it may be permissible to use the ground statistics for budget allocation.

\textbf{(2) Minimize total privacy loss subject to a utility target:}
\begin{subequations}
    \label{eq:2}
    \begin{align}
        \minimize_{\boldsymbol{\varepsilon} \in \mathbb{R}_{\geq 0}^L} &\quad \sum_{\ell=1}^L \varepsilon_\ell \\
        \texttt{s.t.}& \quad  \mathsf{MSE}_{\mathbf{w}}(\boldsymbol{\varepsilon}) \leq \tau.
    \end{align}
\end{subequations}

This formulation applies when the desired level of utility is fixed (e.g., when the data curator specifies a target error threshold $\tau$) and seeks to minimize the total privacy cost. Such scenarios are common in practice when utility guarantees must be met under strict privacy budgets~\cite{apple2017privacy}. The optimization determines how to allocate privacy budgets across levels to meet the accuracy constraint with minimal cumulative privacy loss.

\textbf{Convexity and Optimal Allocation.}\label{Convexity_Opt_alloc}
The objective $\mathsf{MSE}_{\mathbf{w}}(\boldsymbol{\varepsilon})$~\eqref{eq:multi_level_mse} is strictly convex in $\boldsymbol{\varepsilon}$ over $\mathbb{R}_{> 0}^L$. This follows from Proposition~\ref{proposition:MSE_convex}, which establishes the strict convexity of each node-level MSE in $\varepsilon_\ell$ for fixed $N_j$, and the fact that convexity is preserved under nonnegative weighted sums and finite summations. As a result, the optimization problems~\eqref{eq:1} and~\eqref{eq:2} are convex and can be efficiently solved in polynomial time using interior-point (IP) methods.

Furthermore, we establish a key theoretical result showing that, under equal weighting, the optimal privacy budget allocation is monotonic across levels and allocates more budget to lower levels in the hierarchy.
\begin{theorem}
\label{theorem:bottom_heavy_optimal}
Under equal weighting, the optimal allocation of the privacy budget satisfies
\[
\varepsilon_1 \leq \varepsilon_2 \leq \cdots \leq \varepsilon_L.
\]
\end{theorem}
\begin{proof}
    The proof is provided in Appendix~\ref{appendix:proofs}.
\end{proof}

\subsection{Hierarchical Post-Processing}\label{sec:hierarchical_post_processing}
It is often desirable to post-process the privatized data to enforce hierarchical consistency (i.e., the sum of lower-level counts must match the total) since independent noise addition can violate this constraint. Due to post-processing immunity, this adjustment does not compromise privacy. The following program enforces consistency, as also described in~\cite{Kuo_DP_hierarchical_histogram, FvHZ:AIJ21}:
\begin{subequations}
    \label{eq:3}
    \begin{align}
        \minimize_{v \in \mathbb{R}_{\geq 0}^n} &\; \lVert v - \{\tilde{N}_i\}_{i=1}^n \rVert_2 \label{obj:3a} \\
        \texttt{s.t.}& \sum_{i} v_i = \tilde{N}_\text{total} \label{c:3b}.
    \end{align}
\end{subequations}

In multi-level hierarchies, this procedure is applied recursively in a top-down manner: first between levels 1 and 2, then between each level-2 node and its children at level 3, and so on.

\section{Experimental Results}\label{Exp_Section}

\begin{figure}[t]
  \centering
  \includegraphics[width=0.65\textwidth]{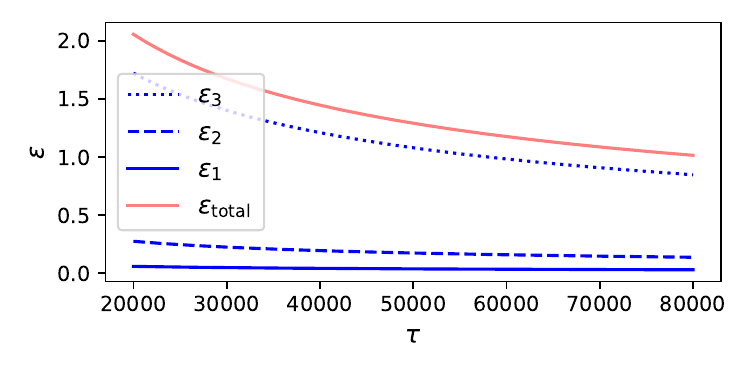}
  \caption{Privacy budget allocation using Optimization Program~\eqref{eq:2} in Wyoming.}
  \label{fig:opt_2nd_56}
\end{figure}

Next, we present empirical evidence demonstrating the effectiveness of the proposed optimization method. We evaluate both bias and variance, as well as the impact of enforcing hierarchical constraints. Experiments are conducted on the release of three-level population statistics, where the hierarchy follows the geographic structure: State $\rightarrow$ Census Tract $\rightarrow$ Census Block.

\textbf{Dataset and Setting.} Our experiments use a subset of the U.S. Census data from the 2020 Census Privacy-Protected Microdata File (PPMF)~\cite{ppmf}, released under the disclosure avoidance methodology adopted for the 2020 Census. Specifically, we use the 2010 release version of the PPMF~\cite{ppmf_2010} to compute the budget allocation that satisfies exact $\varepsilon$-DP.

In the main experiment, we focus on the state of Wyoming (1 state, 128 Census Tracts, and 20{,}975 Blocks), where $\varepsilon_1$, $\varepsilon_2$, and $\varepsilon_3$ denote the privacy budgets for the State, Tract, and Block levels, respectively. \emph{Additional results for other states are provided in the Appendix~\ref{additional_experiments_State}}, showing the same clear trend and reinforcing that the analysis presented in this section applies broadly.

\begin{figure}[t]
  \centering
  \includegraphics[width=0.95\textwidth]{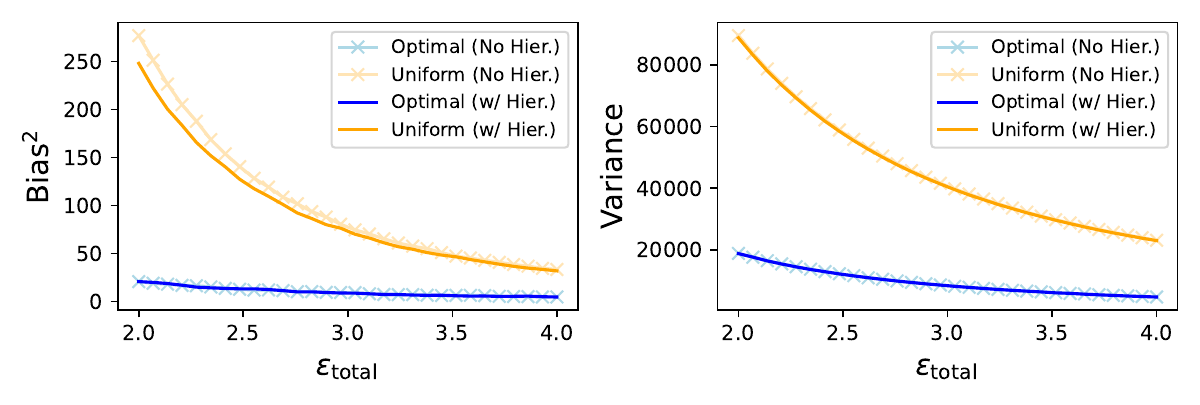}
  \caption{Hierarchical data release performance in Wyoming.}
  \label{fig:1st_3rd_56}
\end{figure}

\textbf{Evaluation Metrics and Baseline.}
We evaluate both $\text{Bias}^2$ and Variance, following the standard decomposition of mean squared error (MSE). Our baseline is uniform allocation of the total privacy budget, a commonly suggested heuristic for hierarchical data structures~\cite{Cohen_DP_Redistricting_2021}. For example, in this setting, a total privacy budget of $\varepsilon_\text{total}$ is divided evenly: $\varepsilon_1 = \varepsilon_2 = \varepsilon_3 = \frac{\varepsilon_\text{total}}{3}$.

\subsection{Optimized Allocation: Bias$^2$ and Variance}\label{optimized_alloc_bias_var_section}

We begin by comparing the performance of the DP-post-processed estimates under a non-negativity constraint, without enforcing hierarchical consistency, as shown in Figure~\ref{fig:1st_3rd_56} under ``No~Hier." The optimization-based approach significantly outperforms uniform allocation in terms of $\text{Bias}^2$, with the uniform method yielding 10 times higher $\text{Bias}^2$. Variance is also approximately four times as high under the uniform approach. This shows that, for the same total privacy budget, both bias and variance can be substantially reduced through optimized budget allocation. Figure~\ref{fig:epsilon_dist_56} illustrates the resulting privacy budget distribution from our proposed method~\eqref{eq:1}, exhibiting the relationship $\varepsilon_1 \leq \varepsilon_2 \leq \varepsilon_3$ as shown by Theorem~\ref{theorem:bottom_heavy_optimal}.

\begin{wrapfigure}{r}{0.45\textwidth}
    \centering
    \vspace{-1em}
    \includegraphics[width=0.45\textwidth]{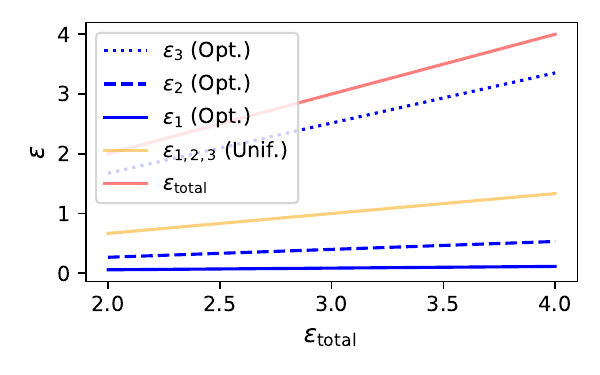}
    \caption{Optimized privacy budget allocation for Wyoming using program~\eqref{eq:1}.}
    \label{fig:epsilon_dist_56}
    \vspace{-1em}
\end{wrapfigure}

Next, we consider the case with hierarchical post-processing, as shown in Figure~\ref{fig:1st_3rd_56} under ``w/~Hier." The overall trend remains similar: the optimized allocation continues to outperform uniform allocation at similar rate. Importantly, both bias and variance are slightly reduced after enforcing hierarchical consistency. By constraining the lower-level population counts to match the upper level, the post-processing step helps mitigate strong positive bias in small counts, which in turn slightly reduces the overall variance.

Finally, we consider the case of minimizing total privacy loss subject to a utility target using program~\eqref{eq:2}, as shown in Figure~\ref{fig:opt_2nd_56}. As the total error tolerance $\tau$ increases, the required total privacy budget $\varepsilon_\text{total}$ decreases. For each value of $\tau$, we observe that the majority of the privacy budget is allocated to the lower level, consistent with Theorem~\ref{theorem:bottom_heavy_optimal}. Despite using 2010 PPMF data to predict the privacy budget required for the 2020 release, our estimates closely align with observed utility. For example, at $\tau = 20{,}000$, the optimization allocates approximately $\varepsilon_\text{total} \approx 2$, and in Figure~\ref{fig:1st_3rd_56}, we see that the total MSE (bias$^2$ + variance) at this budget is indeed around 20{,}000.

\subsection{Ablation on Weights}\label{ablation}

To study the effect of weighting, we vary \( w_3 \) in equation~\eqref{eq:multi_level_mse}, the weight assigned to level 3, and measure the resulting MSE at each level. The remaining weights are distributed equally between levels 1 and 2, i.e., \( w_1 = w_2 = \frac{1 - w_3}{2} \). The total privacy budget is fixed at \( \varepsilon = 2.0 \), and all other parameters are held constant.

\textbf{Experimental Result.} 
We begin by analyzing how total and level-wise MSE change as a function of the weight \( w_3 \). As shown in Figure~\ref{fig:weights_experiment_56}, increasing \( w_3 \) leads to a reduction in MSE at level 3 (\( \text{MSE}_3 \)), as more weight on this level increases the allocated privacy budget \( \varepsilon_3 \). Conversely, MSE at levels 1 and 2 increases with \( w_3 \), since their corresponding weights, and thus \( \varepsilon_1 \) and \( \varepsilon_2 \), decrease.
When all three weights are equal (i.e., \( w_1 = w_2 = w_3 = \frac{1}{3} \)), the optimization corresponds to minimizing total MSE. This setting, indicated by the orange vertical line in Figure~\ref{fig:weights_experiment_56}, yields the lowest total error. Deviating from this uniform weighting shifts accuracy toward the prioritized level but increases overall error.

\begin{figure}[t]
  \centering
  \includegraphics[width=0.95\textwidth]{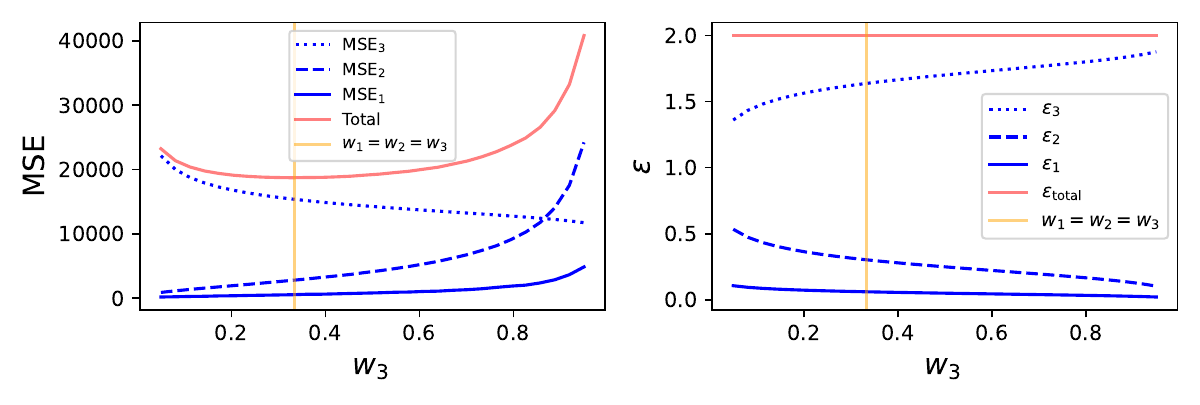}
  \caption{Total and level-wise MSE as a function of varying $w_3$ (left), and privacy budget allocation across levels from program~\eqref{eq:1} under varying $w_3$ (right) in Wyoming.}
  \label{fig:weights_experiment_56}
\end{figure}

\section{Downstream Allocation Task}\label{Downstream_Section}

Motivated by Census applications where hierarchical data inform budget allocation, such as the distribution of Title I funds~\cite{title1} or congressional apportionment~\cite{balinski2001fair}, we analyze the impact of downstream decision-making tasks that rely on privatized hierarchical statistics. 
Many such tasks distribute a fixed budget \( B \) across subgroups according to their size or assessed need. These tasks often incorporate a \textit{weight function} to encode policy preferences; for example, sublinear weighting may be used to prioritize minority populations. We formalize this setting using a generalized allocation framework that captures a broad class of weighted budget rules.

\textbf{Generalized budget allocation task.} 
The process of budget allocation can be broadly described by the following steps:
\begin{enumerate}[leftmargin=*]
    \item \emph{Estimate the size} of each group \( i \in R_2 \) relative to the total population: 
    \(
    \displaystyle P_i \triangleq \frac{N_i}{\sum_{j \in R_2} N_j} = \frac{N_i}{N_\text{total}}.
    \)
    
    \item \emph{Apply a weight function} \( W: \mathbb{R} \rightarrow \mathbb{R} \) to each group's proportion to reflect policy preferences:
    \(
    \displaystyle W_i \triangleq W(P_i).
    \)
    
    \item \emph{Distribute} the total budget \( B \) proportionally according to the computed weights.
\end{enumerate}
If \( W \) is the identity function (i.e., \( W(p) = p \)), the allocation is directly proportional to group sizes.

\subsection{Bias$^2$ and Variance under Different Weight Functions}

We study the Bias$^2$ and Variance of the allocation derived from privatized weights \( \tilde{W}_i \) (i.e., using privatized statistics to compute proportions and apply the weight function) and assess how the choice of weight function \( W \) influences downstream performance. Specifically, we evaluate the bias and variance of the resulting \textit{misallocation} as a percentage, rather than in terms of a fixed total budget \( B \).

We evaluate three weight functions: \emph{logarithmic} (which favors minority groups), \emph{linear} (which is proportional), and \emph{quadratic} (which favors majority groups), to examine how allocation preferences affect performance under different privacy mechanisms.

\textbf{Dataset and Setting.}
We use the same dataset as in Section~\ref{Exp_Section}, where the 2010 PPMF is used as prior data and the 2020 PPMF serves as ground truth. All experiments are conducted on the state of Wyoming, with additional results for other states provided in Appendix~\ref{additional_experiments_State}. In this experiment, we randomly select a Census Tract within Wyoming and allocate the privacy budget among the Census Blocks within that tract. The procedure involves: \textbf{(1)} allocating the privacy budget using program~\eqref{eq:1}, \textbf{(2)} applying hierarchical post-processing for consistency using program~\eqref{eq:3}, and \textbf{(3)} computing weighted group proportions based on the resulting privatized counts.

\begin{figure}[t]
  \centering
  \includegraphics[width=0.95\textwidth]{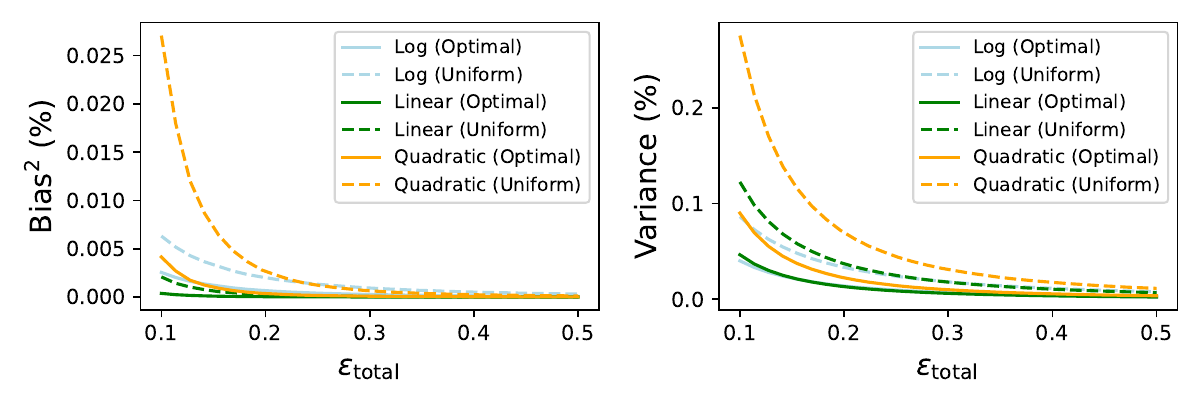}
  \caption{Bias$^2$ and Variance for three weight functions $W(p)$: logarithmic ($\log(p+1)$), linear ($p$), and quadratic ($p^2$), in a single Census Tract in Wyoming.}
  \label{fig:allocation_56_100}
\end{figure}

\textbf{Experimental Results.} 
We begin by evaluating performance under \textit{linear weighting}, where the optimal allocation using Program~\ref{eq:1} significantly reduces both bias and variance. This correction arises from better representation of group sizes in the allocation, reducing misallocation error.

For \textit{nonlinear weight functions}, bias will shift by Jensen’s inequality: since the weight function \( W(\cdot) \) is applied to noisy proportions, we have
\[
\mathbb{E}[W(\tilde{P}_i)] \neq W(\mathbb{E}[\tilde{P}_i]),
\]
which introduces bias. For \textit{convex} functions such as the quadratic, this leads to an positive bias in expected weights; for \textit{concave} functions such as the logarithmic, it results in negative bias. As a result, \(\text{Bias}^2\) is generally larger under nonlinear weighting than under linear.

Under \textit{log weighting}, group size differences are compressed, promoting more equal allocation. Since uniform allocation ignores group sizes and tends to overestimate small populations due to higher noise, it performs more similarly to the optimal allocation under log weighting than under linear or quadratic, resulting in a smaller performance gap.

Regarding variance, log weighting tends to dampen the impact of noise, often resulting in lower variance. In contrast, quadratic weighting amplifies group size differences and can lead to increased variance. However, the behavior of variance is also influenced by the \textit{skewness} of the population. In highly skewed populations, this contrast becomes less apparent. The log transformation flattens group differences, which makes small, noisy groups more influential after normalization, potentially increasing variance. Conversely, the quadratic transformation exaggerates group differences, allowing large, stable groups to dominate the normalization, which can reduce variance in the final allocation.

Overall, across all weight functions, \emph{optimal allocation consistently outperforms uniform allocation}, although the degree of improvement varies with the weighting scheme and population structure.

These findings also highlight the importance of choosing a weighting scheme that aligns with the intended policy objective, as different schemes impact bias and variance in distinct ways. This suggests that practitioners should choose weight functions with careful attention to the fairness objectives or utility demands of their specific application.

\section{Conclusion}\label{Conclusion}

This paper addresses the fundamental challenge of optimal privacy budget allocation for hierarchical data release under differential privacy. We propose a convex optimization framework that minimizes the total \emph{mean squared error} (MSE), accounting for both bias and variance introduced by Laplace noise and non-negativity post-processing. Our theoretical analysis reveals that optimal allocations are bottom-heavy, assigning more of the privacy budget to lower levels of the hierarchy. This strategy is shown to be strictly superior to uniform allocation, both analytically and empirically.

Through extensive experiments on real U.S. Census data, the paper shows that the proposed allocation significantly reduces both bias and variance compared to a standard approach. Additionally, we extend our framework to evaluate downstream resource allocation tasks, showing that optimal allocation improves utility across a range of preference-weighted schemes, highlighting its practical and societal relevance in applications such as public funding and policy-driven resource distribution.

These findings emphasize the importance of structured privacy budget allocation in practical applications of differential privacy, especially in scenarios where hierarchical consistency and data utility are critical. Future directions include extending the approach to other privatization mechanisms, incorporating fairness constraints, and exploring its impact on broader policy-driven decisions.

\bibliographystyle{unsrt}
\bibliography{references}

\appendix

\newpage

\section{Missing Proofs}\label{appendix:proofs}

\begin{proof}[Proof of Proposition \ref{proposition:bias}]
\label{proof:bias}
Let \( f(z) = \frac{\varepsilon}{2} \exp\left( -\varepsilon |z - N| \right) \) be the PDF of the Laplace centered at \( N \) with scale parameter \( 1/\varepsilon \). The expected value of the post-processed count is given by:
\begin{align}
\mathbb{E}[\tilde{N}] &= \int_{-\infty}^{\infty} \max(0, z) f(z) \, dz \nonumber \\
&= \int_{-\infty}^{0} 0 \cdot f(z) \, dz + \int_{0}^{N} z f(z) \, dz + \int_{N}^{\infty} z f(z) \, dz. \label{eq:exp_postproc} 
\end{align}

We now compute each term in \eqref{eq:exp_postproc} separately:
\begin{align}
\int_{-\infty}^{0} 0 \cdot f(z) \, dz &= 0, \label{eq:int1} \\
\int_{0}^{N} z f(z) \, dz &= \frac{1}{2}(N - \tfrac{1}{\varepsilon}) + \frac{1}{2\varepsilon} \exp(-\varepsilon N), \label{eq:int2} \\
\int_{N}^{\infty} z f(z) \, dz &= \frac{1}{2}(N + \tfrac{1}{\varepsilon}). \label{eq:int3}
\end{align}

Combining equations~\eqref{eq:int1}--\eqref{eq:int3}, we obtain:
\[
\mathbb{E}[\tilde{N}] = N + \frac{1}{2\varepsilon} \exp(-\varepsilon N).
\]

Thus, the bias of \( \tilde{N} \) is:
\[
\mathrm{Bias}(\tilde{N}) = \mathbb{E}[\tilde{N}] - N = \frac{1}{2\varepsilon} \exp(-\varepsilon N) > 0.
\]
\end{proof}

\begin{proof}[Proof of Proposition \ref{proposition:variance}]
\label{proof:variance}

From Proposition~\ref{proposition:bias}, the expected value is:
\[
\mathbb{E}[\tilde{N}] = N + \frac{1}{2\varepsilon} e^{-\varepsilon N}.
\]

Let the PDF of the Laplace centered at \( N \) with the scale of $1/\varepsilon$ be:
\[
f(z) = \frac{\varepsilon}{2} \exp(-\varepsilon |z - N|).
\]

We split the second moment:
\[
\mathbb{E}[\tilde{N}^2] = \int_{-\infty}^{0} 0^2 f(z)\,dz + \int_0^{N} z^2 f(z)\,dz + \int_{N}^{\infty} z^2 f(z)\,dz.
\]

\textbf{Region 1:} \( z \in [0, N] \)
\[
\int_0^{N} z^2 f(z)\,dz = \frac{\varepsilon}{2} \int_0^{N} z^2 e^{-\varepsilon (N - z)} dz
= \frac{\varepsilon}{2} e^{-\varepsilon N} \int_0^N z^2 e^{\varepsilon z} dz.
\]

Let \( u = \varepsilon z \), so \( dz = \frac{du}{\varepsilon} \), the limits become \( 0 \) to \( \varepsilon N \):
\[
= \frac{1}{2\varepsilon^2} e^{-\varepsilon N} \int_0^{\varepsilon N} u^2 e^u du.
\]

Using the identity:
\[
\int_0^x u^2 e^u du = (x^2 - 2x + 2)e^x - 2,
\]
we get:
\[
\int_0^N z^2 f(z)\,dz = \frac{1}{2\varepsilon^2} \left( \varepsilon^2 N^2 - 2\varepsilon N + 2 - 2 e^{-\varepsilon N} \right).
\]

\textbf{Region 2:} \( z \in [N, \infty) \)

Let \( u = z - N \), so:
\[
\int_N^\infty z^2 f(z)\,dz = \frac{\varepsilon}{2} \int_0^\infty (u + N)^2 e^{-\varepsilon u} du = \frac{\varepsilon}{2} \int_0^\infty (u^2 + 2Nu + N^2) e^{-\varepsilon u} du .
\]

Expanding:
\[
= \frac{\varepsilon}{2} \left( \int_0^\infty u^2 e^{-\varepsilon u} du + 2N \int_0^\infty u e^{-\varepsilon u} du + N^2 \int_0^\infty e^{-\varepsilon u} du \right).
\]

Using standard integrals:
\[
\int_0^\infty e^{-\varepsilon u} du = \frac{1}{\varepsilon}, \quad
\int_0^\infty u e^{-\varepsilon u} du = \frac{1}{\varepsilon^2}, \quad
\int_0^\infty u^2 e^{-\varepsilon u} du = \frac{2}{\varepsilon^3},
\]
we get:
\[
= \frac{1}{2\varepsilon^2} \left( 2 + 2\varepsilon N + \varepsilon^2 N^2 \right).
\]

\textbf{Combine both parts:}
\[
\mathbb{E}[\tilde{N}^2] = \frac{1}{\varepsilon^2} \left( \varepsilon^2 N^2 + 2 - e^{-\varepsilon N} \right).
\]

Now compute the variance:
\[
\mathrm{Var}(\tilde{N}) = \mathbb{E}[\tilde{N}^2] - \left( \mathbb{E}[\tilde{N}] \right)^2.
\]

Recall:
\[
\mathbb{E}[\tilde{N}] = N + \frac{1}{2\varepsilon} e^{-\varepsilon N} \Rightarrow
\left( \mathbb{E}[\tilde{N}] \right)^2 = N^2 + \frac{1}{\varepsilon} N e^{-\varepsilon N} + \frac{1}{4\varepsilon^2} e^{-2\varepsilon N}.
\]

Therefore:
\[
\begin{aligned}
\mathrm{Var}(\tilde{N}) 
&= \frac{1}{\varepsilon^2} \left( \varepsilon^2 N^2 + 2 - e^{-\varepsilon N} \right)
- \left( N^2 + \frac{1}{\varepsilon} N e^{-\varepsilon N} + \frac{1}{4\varepsilon^2} e^{-2\varepsilon N} \right) \\
&= \frac{1}{\varepsilon^2} \left( 2 - e^{-\varepsilon N} \right)
- \frac{N}{\varepsilon} e^{-\varepsilon N}
- \frac{1}{4\varepsilon^2} e^{-2\varepsilon N}.
\end{aligned}
\]
\end{proof}

\begin{proof}[Proof of Proposition \ref{proposition:MSE_convex}]
Let
\[
f(\varepsilon) = \frac{1}{\varepsilon^2} \left( 2 - e^{-\varepsilon N} \right)
- \frac{N}{\varepsilon} e^{-\varepsilon N}.
\]

We show that \( f''(\varepsilon) > 0 \) for all \( \varepsilon > 0 \). First, compute the first derivative:
\[
f'(\varepsilon) = \frac{N^2 e^{-\varepsilon N} \varepsilon^2 + 2N e^{-\varepsilon N} \varepsilon + 2 e^{-\varepsilon N} - 4}{\varepsilon^3}.
\]
Then, compute the second derivative:
\[
f''(\varepsilon) = -\frac{N^3 e^{-\varepsilon N} \varepsilon^3 + 3\left(N^2 e^{-\varepsilon N} \varepsilon^2 + 2N e^{-\varepsilon N} \varepsilon + 2 e^{-\varepsilon N} - 4\right)}{\varepsilon^4}.
\]

Let
\[
g(\varepsilon) = -N^3 e^{-\varepsilon N} \varepsilon^3 - 3\left(N^2 e^{-\varepsilon N} \varepsilon^2 + 2N e^{-\varepsilon N} \varepsilon + 2 e^{-\varepsilon N} - 4\right),
\]
so that \( f''(\varepsilon) = \frac{g(\varepsilon)}{\varepsilon^4} \). It suffices to show \( g(\varepsilon) > 0 \) for all \( \varepsilon > 0 \).

First, evaluate the function at \( \varepsilon = 0 \). Taking the limit,
\[
\lim_{\varepsilon \to 0} g(\varepsilon) = -0 - 3(0 + 0 + 2 - 4) = 6.
\]
Next, compute the derivative:
\[
g'(\varepsilon) = N^4 e^{-\varepsilon N} \varepsilon^3 > 0 \quad \text{for all } \varepsilon > 0.
\]
Since \( g(0) = 6 > 0 \) and \( g'(\varepsilon) > 0 \), it follows that \( g(\varepsilon) > 0 \) for all \( \varepsilon > 0 \).

Therefore, \( f''(\varepsilon) = \frac{g(\varepsilon)}{\varepsilon^4} > 0 \) for all \( \varepsilon > 0 \), proving that \( f \) is strictly convex on \( \mathbb{R}_{> 0} \).
\end{proof}

\begin{proof}[Proof of Proposition \ref{proposition:MSE_bounded}]
Let
\[
f(N) = \frac{1}{\varepsilon^2}(2 - e^{-\varepsilon N}) - \frac{N}{\varepsilon} e^{-\varepsilon N}.
\]

First, evaluate the lower bound at \( N = 0 \):
\[
f(0) = \frac{1}{\varepsilon^2}(2 - 1) - \frac{0}{\varepsilon} \cdot 1 = \frac{1}{\varepsilon^2}.
\]

Now consider the limit as \( N \to \infty \):
\[
\lim_{N \to \infty} f(N) = \frac{1}{\varepsilon^2} \cdot 2 = \frac{2}{\varepsilon^2}.
\]

To show monotonicity, compute the derivative:
\[
\frac{d}{dN} f(N) = N e^{-\varepsilon N}.
\]

This is non-negative for all \( N \geq 0 \), and strictly positive for \( N > 0 \), which implies that \( f(N) \) is strictly increasing in \( N \).

Hence, the function is bounded below by its value at \( N = 0 \) and bounded above by its limit as \( N \to \infty \). Therefore,
\[
\frac{1}{\varepsilon^2} \leq f(N) < \frac{2}{\varepsilon^2}.
\]
\end{proof}

\begin{proof}[Proof of Theorem \ref{theorem:bottom_heavy_optimal}]
To show that $\varepsilon_1 \leq \varepsilon_2 \leq \cdots \leq \varepsilon_L $, it suffices to show that $\varepsilon_n \leq \varepsilon_{n+1}$ for any two consecutive levels in the hierarchy. Consider an arbitrary two-level hierarchy where the upper-level node has count $N$ and the lower-level consists of subcounts $N_i$ satisfying $\sum_i N_i = N$.

Define the per-node MSE as a function of $\varepsilon > 0$ and count $N \geq 0$:
\[
f(\varepsilon, N) = \frac{1}{\varepsilon^2}(2 - e^{-\varepsilon N}) - \frac{N}{\varepsilon} e^{-\varepsilon N}.
\]

Suppose, for the sake of contradiction, that $\varepsilon_n > \varepsilon_{n+1}$ is optimal. Then, the marginal utility of privacy budget must be equal across levels:
\[
\frac{\partial f(\varepsilon_n, N)}{\partial \varepsilon} = \sum_i \frac{\partial f(\varepsilon_{n+1}, N_i)}{\partial \varepsilon}.
\]

Let us define:
\[
g(\varepsilon, N) \triangleq \frac{\partial f(\varepsilon, N)}{\partial \varepsilon} = \frac{N^2 e^{-\varepsilon N} \varepsilon^2 + 2N e^{-\varepsilon N} \varepsilon + 2 e^{-\varepsilon N} - 4}{\varepsilon^3}.
\]

First, observe that $g(\varepsilon, N)$ is strictly increasing in $\varepsilon$ since $f(\varepsilon, N)$ is strictly convex in $\varepsilon$ by Proposition~\ref{proposition:MSE_convex}. Therefore, for fixed $N$, we have:
\[
g(\varepsilon_n, N) > g(\varepsilon_{n+1}, N).
\]

Next, notice that $\frac{\partial g(\varepsilon, N)}{\partial N} = -N^2 e^{-\varepsilon N} \leq 0$, so $g(\varepsilon, N)$ is monotonically decreasing in $N$.

Furthermore, observe:
\[
g(\varepsilon, 0) = -\frac{2}{\varepsilon^3}, \quad \lim_{N \to \infty} g(\varepsilon,N) = -\frac{4}{\varepsilon^3}.
\]
Hence, for all $N \geq 0$, we have the bound:
\[
-\frac{4}{\varepsilon^3} \leq g(\varepsilon,N) \leq -\frac{2}{\varepsilon^3}.
\]

Let $k$ be the number of nodes at the lower level. Now, consider two cases:  
\textbf{(1)} If $k = 1$, then trivially $g(\varepsilon, N) = \sum_i g(\varepsilon, N_i)$, since there is only a single lower-level node and it is equal to $N$.  
\textbf{(2)} If $k \geq 2$, then using the bounds on $g(\varepsilon, N)$, we have:
$$
g(\varepsilon, N) \geq -\frac{4}{\varepsilon^3} \geq k \cdot \left(-\frac{2}{\varepsilon^3}\right) \geq \sum_i g(\varepsilon, N_i),
$$
since each $g(\varepsilon, N_i) \leq -\frac{2}{\varepsilon^3}$ and there are $k$ such terms in the sum.

Combining the two observations, we obtain
\[
g(\varepsilon_n, N) > g(\varepsilon_{n+1}, N) \geq \sum_i g(\varepsilon_{n+1}, N_i),
\] which is a contradiction.
\end{proof}

\FloatBarrier
\section{Skewness Analysis}\label{skewness_analysis}

\begin{figure}[h]
    \centering
    \begin{minipage}{0.47\textwidth}
        \centering
        \includegraphics[width=\linewidth]{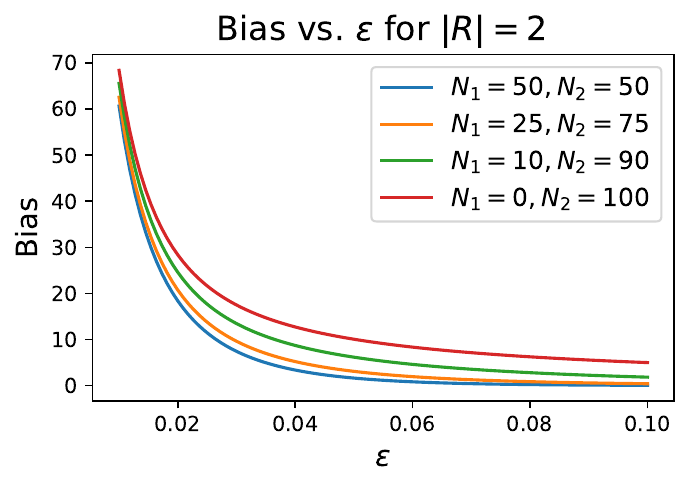}
    \end{minipage}
    \hfill
    \begin{minipage}{0.48\textwidth}
        \centering
        \includegraphics[width=\linewidth]{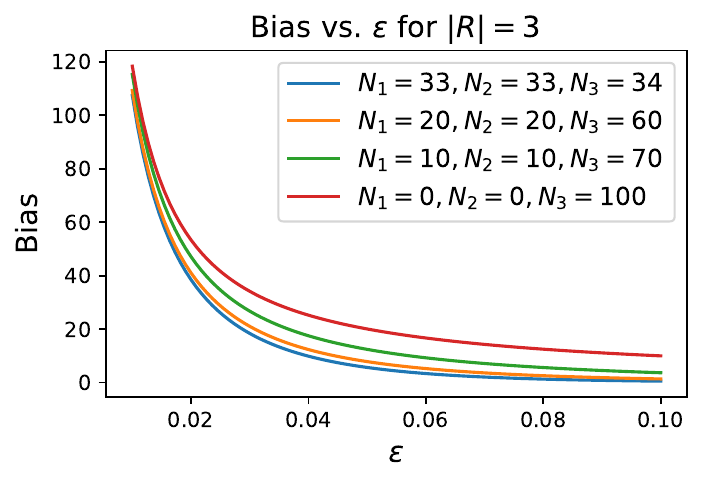}
    \end{minipage}
    \caption{Bias incurred as a function of $\varepsilon$ for different distributions of $N_i$, where the total sum is fixed at 100. The left plot represents the 2-region case ($N_1 + N_2 = 100$), and the right plot represents the 3-region case ($N_1 + N_2 + N_3 = 100$).}
    \label{fig:skewness_plot}
\end{figure}

In this section, we investigate how the skewness of a population distribution affects bias. Even when the total population and the number of regions are fixed, the bias introduced by non-negativity post-processing~\eqref{non-neg_post-processing} can vary depending on how the population is distributed. As illustrated in Figure~\ref{fig:skewness_plot}, bias decreases as the distribution becomes more uniform. In fact, the bias is minimized when the population is evenly distributed across all regions.

\begin{theorem}
For a fixed privacy budget $\varepsilon$, the bias is minimized when $N_i$ is uniform for all $i \in R$.
\end{theorem}

\begin{proof}
For a fixed $\varepsilon$, total population $N$, and number of regions $|R|$, minimizing the bias is equivalent to solving the following optimization problem:
\begin{subequations}
\begin{align} 
    \minimize_{N_i} & \sum_{i \in R} e^{-N_i} \\
    \text{s.t.} & \sum_{i \in R} N_i = N, \\
    & N_i \geq 0, \quad \forall i \in R.
\end{align}
\end{subequations}

The Lagrangian for this problem is given by:
\[
\mathcal{L}(N_i, \lambda, \mu_i) = \sum_{i \in R} e^{-N_i} + \lambda \left(\sum_{i \in R} N_i - N \right) - \sum_{i \in R} \mu_i N_i.
\]

The necessary KKT conditions are:

\begin{enumerate}
    \item \emph{Stationarity Condition}
    \[
    \frac{\partial \mathcal{L}}{\partial N_i} = -e^{-N_i} + \lambda - \mu_i = 0.
    \]
    Rearranging, we obtain:
    \[
    \lambda - \mu_i = e^{-N_i}.
    \]
    
    \item \emph{Primal Feasibility}
    \[
        \sum_{i \in R} N_i = N, \quad N_i \geq 0, \quad \forall i \in R.
    \]
    
    \item \emph{Dual Feasibility}
    \[
        \mu_i \geq 0, \quad \forall i \in R.
    \]
    
    \item \emph{Complementary Slackness}
    \[
        \mu_i N_i = 0, \quad \forall i \in R.
    \]
\end{enumerate}

If $N_i > 0$, then $\mu_i = 0$, so we obtain:
\[
  \lambda = \frac{1}{e^{N_i}}.
\]
Taking the natural logarithm on both sides:
\[
  N_i = -\ln \lambda.
\]

Using the primal feasibility, we have:
\[
    \sum_{i \in R} N_i = \sum_{i \in R} (-\ln \lambda) = |R| (-\ln \lambda) = N.
\]

Solving for $\lambda$:
\[
    \lambda = e^{-N/|R|}.
\]

Substituting back:
\[
    N_i^* = \frac{N}{|R|}, \quad \forall i \in R.
\]

Since the objective function is strictly convex over the feasible set, the optimal solution must be unique. This implies that the only possible minimizer is:
\[
    N_i^* = \frac{N}{|R|}, \quad \forall i \in R.
\]
Thus, the optimal allocation distributes $N$ evenly among all states  $i \in R$.
\end{proof}

This also implies that if two distributions of $N_i$, written in vector form as $\{N_i^B\}_{i \in R}$ and $\{N_i^A\}_{i \in R}$, satisfy $\{N_i^B\}$ \textit{majorizes} $\{N_i^A\}$, then the total bias is higher under $\{N_i^B\}$. Since the function $f(x) = e^{-x}$ is convex, Karamata's inequality gives
\[
\sum_{i \in R} f(N_i^A) \leq \sum_{i \in R} f(N_i^B),
\]
which implies greater cumulative bias in more skewed distributions.

\FloatBarrier
\section{Limitations}\label{limitation}

We discuss two primary limitations of the paper.

First, the optimization program relies on previously released data to determine the allocation of privacy budget. If level-wise privacy budgets are also made public, then the program~\eqref{eq:1} is not differentially private unless the input counts used for computing the allocation are themselves privatized. However, as discussed in Section~\ref{Exp_Section}, if only the total privacy budget is released—without revealing the per-level allocation—then it is acceptable to use ground-truth statistics for budget allocation.

Second, the hierarchical post-processing method introduced in Section~\ref{sec:hierarchical_post_processing} assumes that lower levels typically contribute more to the total MSE under optimal allocation. While this is generally true, the contribution can vary depending on the level-wise weights. In particular, under extreme weight settings (e.g., $w_3 \approx 0.9$ in Figure~\ref{fig:weights_experiment_56}), the upper level may contribute more to the total MSE. In such cases, a bottom-up projection approach—where higher-level counts are adjusted to match lower-level estimates—may be more appropriate than the top-down method used here.

\FloatBarrier
\section{Additional Experiments on Other States}\label{additional_experiments_State}

\FloatBarrier
\subsection{Additional Result on Section~\ref{optimized_alloc_bias_var_section}}

\begin{figure}[htbp]
    \centering

    % New Mexico
    \begin{subfigure}{\textwidth}
        \centering
        \includegraphics[width=0.95\textwidth]{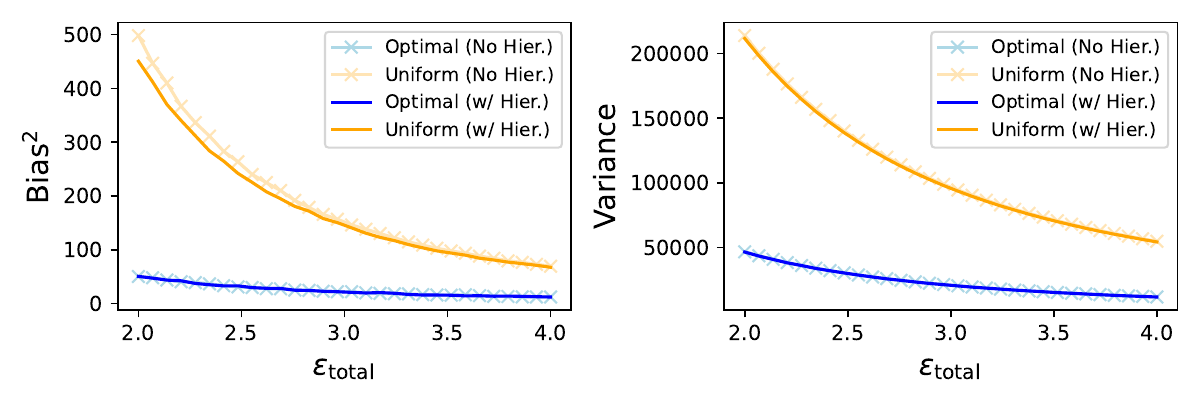}
        \vspace{-2em}
        \caption{New Mexico}
    \end{subfigure}

    % Delaware
    \begin{subfigure}{\textwidth}
        \centering
        \includegraphics[width=0.95\textwidth]{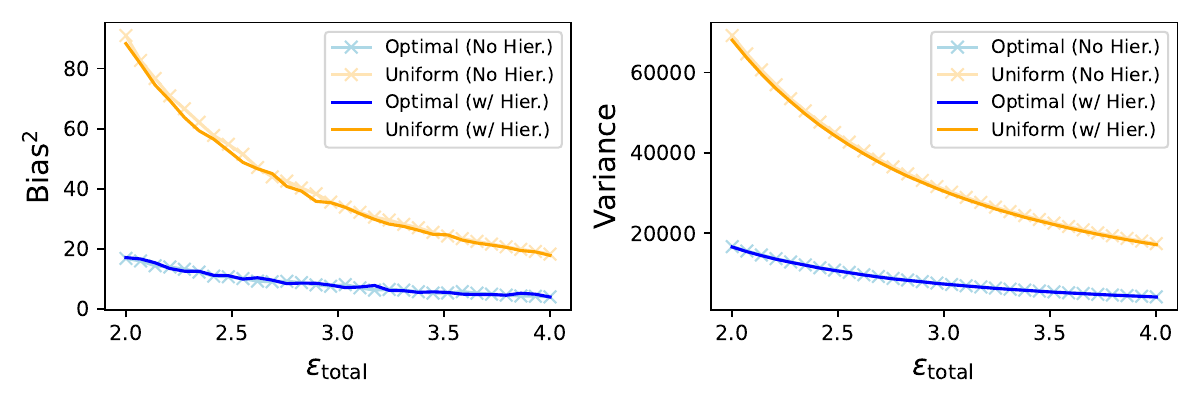}
        \vspace{-2em}
        \caption{Delaware}
    \end{subfigure}
    
    \caption{Hierarchical data release performance in \textbf{New Mexico} (top) and \textbf{Delaware} (bottom). Each row shows bias (left) and variance (right) for three-level hierarchical releases.}
    \label{fig:1st_3rd_appendix}
\end{figure}

\begin{figure}[htbp]
    \centering
    \begin{subfigure}{0.48\textwidth}
        \centering
        \includegraphics[width=\linewidth]{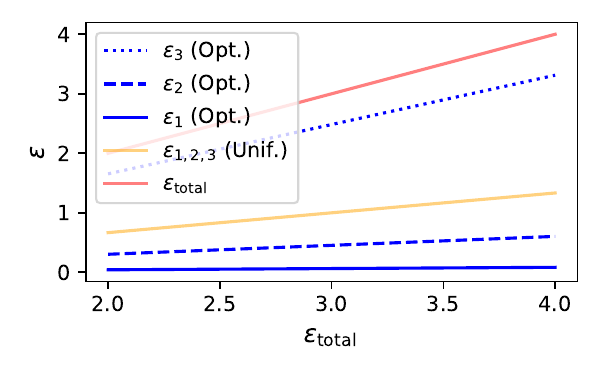}
        \vspace{-2em}
        \caption{New Mexico}
        % \label{fig:allocation_nm}
    \end{subfigure}
    \hfill
    \begin{subfigure}{0.48\textwidth}
        \centering
        \includegraphics[width=\linewidth]{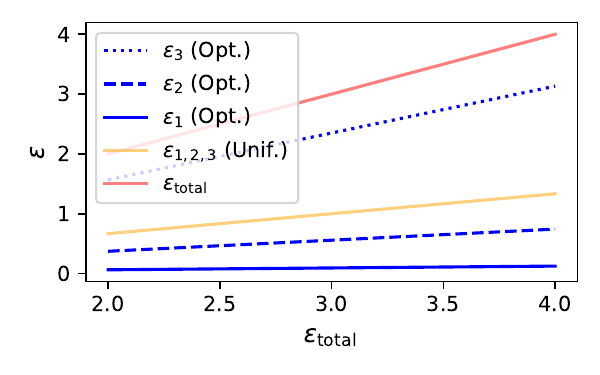}
        \vspace{-2em}
        \caption{Delaware}
        % \label{fig:allocation_de}
    \end{subfigure}
    \caption{Bias$^2$ and variance for three weight functions $W(p)$—logarithmic ($\log(p + 1)$), linear ($p$), and quadratic ($p^2$)—in a single Census Tract in \textbf{New Mexico} (left) and \textbf{Delaware} (right).}
    \label{fig:allocation_35_107}
\end{figure}

\begin{figure}[htbp]
  \centering
  \begin{subfigure}{0.7\textwidth}
    \centering
    \includegraphics[width=\linewidth]{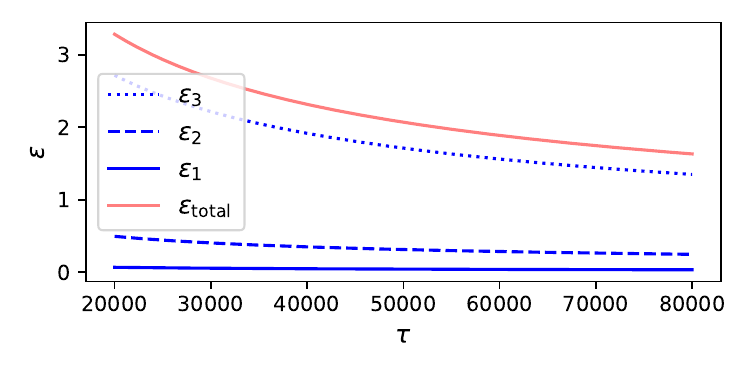}
    \vspace{-2em}
    \caption{New Mexico}
    \label{fig:opt_2nd_35}
  \end{subfigure}
  
  \begin{subfigure}{0.72\textwidth}
    \centering
    \includegraphics[width=\linewidth]{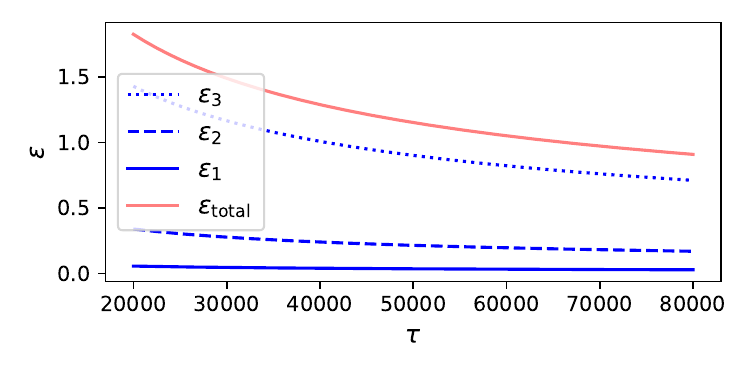}
    \vspace{-2em}
    \caption{Delaware}
    \label{fig:opt_2nd_10}
  \end{subfigure}
  
  \caption{Privacy budget allocation using Optimization Program~\eqref{eq:2} for \textbf{New Mexico} (top) and \textbf{Delaware} (bottom).}
  \label{fig:opt_2nd_combined}
\end{figure}

\FloatBarrier
\subsection{Additional Result on Section~\ref{Downstream_Section}}

\begin{figure}[htbp]
    \centering

    \begin{subfigure}{\textwidth}
        \centering
        \includegraphics[width=0.95\textwidth]{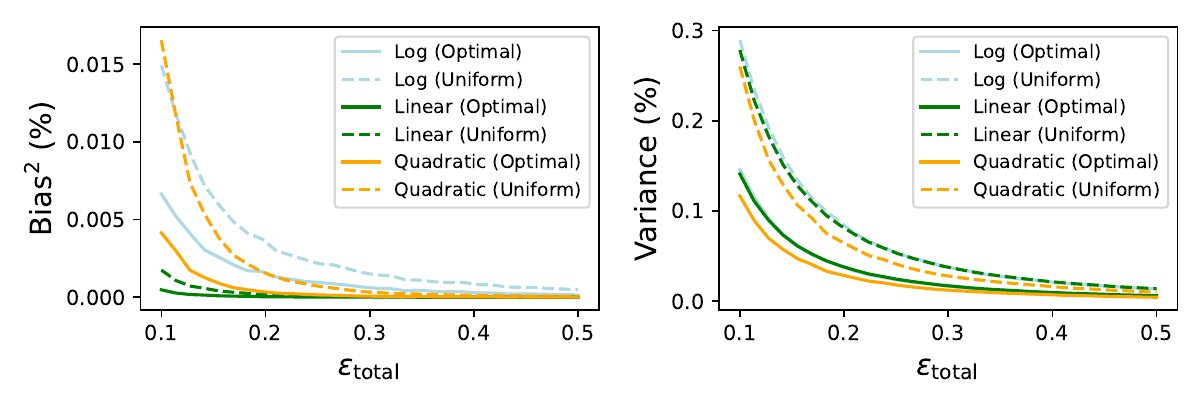}
        \vspace{-1em}
        \caption{New Mexico}
        % \label{fig:allocation_nm}
    \end{subfigure}

    \begin{subfigure}{\textwidth}
        \centering
        \includegraphics[width=0.95\textwidth]{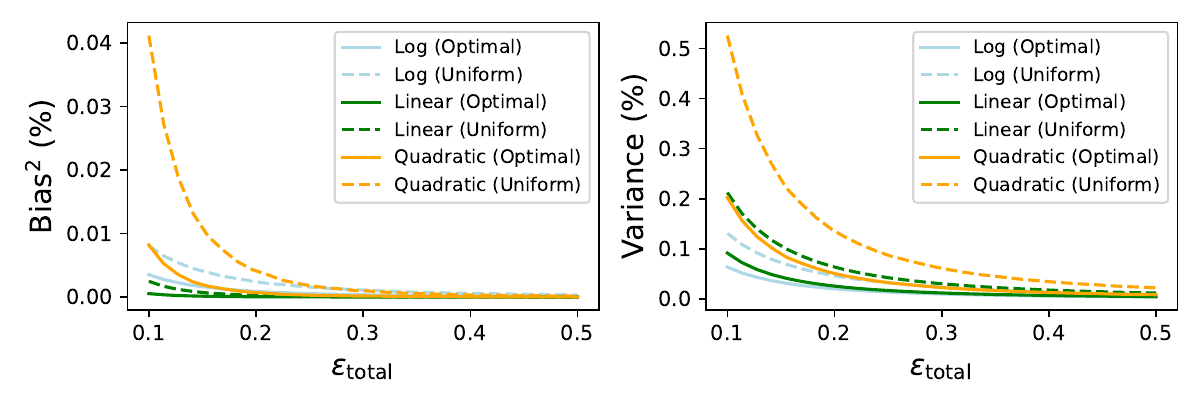}
        \vspace{-1em}
        \caption{Delaware}
        % \label{fig:allocation_de}
    \end{subfigure}

    \caption{Bias$^2$ and variance for three weight functions $W(p)$—logarithmic ($\log(p+1)$), linear ($p$), and quadratic ($p^2$)—in a single Census Tract in \textbf{New Mexico} (top) and \textbf{Delaware} (bottom).}
    \label{fig:allocation_combined_states}
\end{figure}

\end{document}